\documentclass[twoside,12pt]{article}
\usepackage{CJK}
\usepackage{indentfirst, color}
\usepackage{bm}
\usepackage{graphicx}
\usepackage{epsfig}
\usepackage{amsmath}
\usepackage{amsfonts}
\usepackage{amssymb}
\usepackage{amsthm}
\usepackage[superscript]{cite}
\usepackage{latexsym}

\newtheorem{thm}{Theorem}
\newtheorem{lemma}{Lemma}
\newtheorem{corollary}{Corollary}
\usepackage{epsfig}
\usepackage{epstopdf}
\usepackage{pgf,fancyhdr}
\usepackage{float}
\usepackage{amsmath}

\begin{document}

\begin{CJK*}{GBK}{song}

\begin{center}
\LARGE\bf Monogamy and polygamy relations of quantum correlations for multipartite systems
\end{center}

\begin{center}
\rm  Mei-Ming Zhang,$^1$ \  Naihuan Jing,$^{2, 1,*}$  \ and  Hui Zhao,$^3$
\end{center}

\begin{center}
\begin{footnotesize} \sl
$^1$ Department of Mathematics, Shanghai University, Shanghai 200444, China  %

$^2$ Department of Mathematics, North Carolina State University, Raleigh, NC 27695, USA

$^3$ College of Applied Sciences, Beijing University of Technology, Beijing 100124, China

$^*$ Corresponding author: jing@ncsu.edu

\end{footnotesize}
\end{center}

\begin{center}
\begin{minipage}{15.5cm}
\parindent 20pt\footnotesize
We study the monogamy and polygamy inequalities of quantum correlations in arbitrary dimensional multipartite quantum systems. We first derive the monogamy inequality of the $\alpha$th ($0\leq\alpha\leq\frac{r}{2}, r\geq2$) power of concurrence for any $2\otimes2\otimes2^{n-2}$ tripartite states and generalize it to the $n$-qubit quantum states. In addition to concurrence, we show that the monogamy relations are satisfied by other quantum correlation measures
such as entanglement of formation. Moreover, the polygamy inequality of the $\beta$th ($\beta\leq0$) power of concurrence and the $\beta$th ($\beta\geq s, 0\leq s\leq1$) power of the negativity are presented for $2\otimes2\otimes2^{n-2}$. We then obtain the polygamy inequalities of quantum correlations for multipartite states. Finally, our results are shown to be tighter than previous studies using detailed examples.
\end{minipage}
\end{center}

\begin{center}
\begin{minipage}{15.5cm}
\begin{minipage}[t]{2.3cm}{\bf Keywords:}\end{minipage}
\begin{minipage}[t]{13.1cm}
Monogamy, Polygamy, Concurrence, Negativity
\end{minipage}\par\vglue8pt

\end{minipage}
\end{center}

\section{Introduction}
In quantum theory there are different kinds of correlations between multipartite quantum states.
Understanding 
the nature of quantum correlations is a challenge in the study of quantum information processing.
A fundamental property of quantum correlations is that it can be monogamous for multipartite quantum systems.
To be more clearly, if a correlation measure is monogamous, it puts a restriction on the shareability of correlation between the different parties of a composite quantum state.
The monogamy of quantum correlations plays a significant role in quantum world such as the security of quantum cryptography [1].

Considerable efforts 
have been drawn to this task ever since Coffman, Kundu, and Wootters (CKW) presented the first quantitative monogamy relation in [2] for three-qubit states.
It was showed that the squared concurrence [3] and the squared convex-roof extended negativity (CREN) [4,5] satisfied the monogamy relations for multi-qubit states.
Generalized families of monogamy inequalities related to the $\alpha$th ($\alpha\geq2$) power of concurrence [6, 7] and CREN [7-9] tightened the monogamy inequalities, which are useful in featuring multipartite entangled systems in entanglement distributions.
Another important concept is the assisted entanglement, which can be viewed as a dual physical quantity to entanglement measure and gives rise to polygamy relations that are dually monogamous property in multipartite quantum systems.
The authors in [10] provided a characterization of multiqubit entanglement constraints in terms of the assistance of concurrence and negativity.
By using the square of CREN and the Hamming weight of the binary vector related with the distribution of subsystems, the polygamy inequalities in terms of the $\beta$th power of CREN for $0\leq\beta\leq1$ were established in [11].
Tighter polygamy inequalities for the $\beta$th ($0\leq\beta\leq1$) power of concurrence [12] and the CREN [13] were investigated.
Although lots of monogamy and polygamy relations were proposed, there were only a few monogamy relations for the $\alpha$th ($0\leq\alpha\leq1$) power of entanglement measures and polygamy relations of the $\beta$th power for $\beta\geq1$ [14, 15].
Besides concurrence and negativity, unified-$(q,s)$ entanglement also have monogamy and polygamy relations.
Using unified-$(q,s)$ entanglement with real parameters $q$ and $s$, a class of monogamy and polygamy inequalities of multi-qubit systems were proposed [14-20].

In this paper, our main goal is to present unified and tightened monogamy and polygamy relations of entanglement measures including the concurrence, the negativity and the entanglement of formation for multipartite systems. In Section 2, we present the monogamy inequalities of concurrence for $2\otimes2\otimes2^{n-2}$ and $n$-qubit quantum states. We generalize the monogamy relation to the  measures of entanglement for the multipartite quantum systems. By detailed examples, our results are seen to be superior to the previously published results. In Section 3, we derive the polygamy inequalities of concurrence and negativity for $2\otimes2\otimes2^{n-2}$ and $n$-qubit quantum systems. We generalize the polygamy inequalities to the quantum correlation measures for the multipartite quantum systems. We also give an example to show that our bound is tighter than previous available results. Comments and conclusions are given in Section 4.

\section{Tighter monogamy relations of quantum correlations}
Let $H_{A_1}$ and $H_{A_2}$ be $d_{A_1}$- and $d_{A_2}$-dimensional Hilbert spaces. The concurrence of a bipartite quantum pure state $|\varphi\rangle_{A_1A_2}\in H_{A_1}\otimes\ H_{A_2}$ is defined by [21]
\begin{eqnarray}
C(|\varphi\rangle_{A_1A_2})=\sqrt{2[1-tr(\rho_{A_1}^2)]},
\end{eqnarray}
where $\rho_{A_1}$ is the reduced density matrix of $\rho=|\varphi\rangle_{A_1A_2}\langle\varphi|$, i.e, $\rho_{A_1}=tr_{A_2}(\rho)$. For a mixed bipartite quantum state $\rho_{A_1A_2}=\sum_ip_i|\varphi_i\rangle_{A_1A_2}\langle\varphi_i|\in H_{A_1}\otimes H_{A_2}$, the concurrence is given by the convex roof
\begin{eqnarray}
C(\rho_{A_1A_2})=\min_{\{p_i,|\varphi_i\rangle\}}\sum_ip_iC(|\varphi_i\rangle_{A_1A_2}),
\end{eqnarray}
where the minimum is taken over all possible convex partitions of $\rho_{A_1A_2}$ into pure state ensembles $\{p_i,|\varphi_i\rangle\}$, $0\leq p_i\leq1$ and $\sum_ip_i=1$.

Remarkably, for a 2-qubit mixed state $\rho$, the concurrence of $\rho$ is given by [2]
\begin{eqnarray}
C(\rho)=max\{\lambda_1-\lambda_2-\lambda_3-\lambda_4, 0\},
\end{eqnarray}
where $\lambda_{i}$, $i=1,\cdots,4$, denote the square roots of nonnegative eigenvalues of the matrix $\rho(\sigma_y\otimes\sigma_y)\rho^*(\sigma_y\otimes\sigma_y)$ in descending order, $\sigma_y$ is the Pauli matrix, and $\rho^*$ denotes the complex conjugate of $\rho$.

For an $n$-qubit quantum states $\rho_{A_1|A_2A_3\cdots A_n}$, regarded as a bipartite state under bipartite partition $A_1|A_2A_3\cdots A_n$, the concurrence satisfies the monogamy inequality for $\alpha\geq2$ [6]:
\begin{eqnarray}
C^\alpha(\rho_{A_1|A_2A_3\cdots A_n})\geq C^\alpha(\rho_{A_1A_2})+C^\alpha(\rho_{A_1A_3})+\cdots+C^\alpha(\rho_{A_1A_n}),
\end{eqnarray}
where $\rho_{A_1A_i}$, $i=2,\cdots,n$, are the reduced density matrices of $\rho$, i.e, $\rho_{A_1A_i}=tr_{A_2\cdots A_{i-1}A_{i+1}\cdots A_n}(\rho)$.

\begin{lemma}\label{lemma:1}
For real numbers $k\geq1$, $t\geq k$ and $\frac{1}{2}\leq p\leq1$, we have for $0\leq x\leq\frac{1}{2}$:
\begin{eqnarray}
(1+t)^x\geq p^x+\frac{(1+k)^x-p^x}{k^x}t^x.
\end{eqnarray}
\end{lemma}

\begin{proof}
Let $g(x,y)=(1+\frac{1}{y})^{x-1}-p^x$ where $\frac{1}{2}\leq p\leq1$, $0\leq x\leq\frac{1}{2}$ and $0<y\leq\frac{1}{k}$ with $k\geq1$.
Then $\frac{\partial g}{\partial x}=(1+\frac{1}{y})^{x-1}ln(1+\frac{1}{y})-p^xlnp>0$ as $1+\frac{1}{y}\geq2$. Therefore, $g(x,y)$ is an increasing function of $x$, i.e, $g(x,y)\leq g(\frac{1}{2},y)=(1+\frac{1}{y})^{-\frac{1}{2}}-p^{\frac{1}{2}}\leq0$ as $0<(1+\frac{1}{y})^{-1}\leq\frac{1}{2}$. Let $f(x,y)=(1+y)^x-(py)^x$ with $\frac{1}{2}\leq p\leq1$, $0\leq x\leq\frac{1}{2}$ and $0<y\leq\frac{1}{k}$. Since $(1+\frac{1}{y})^{x-1}-p^x\leq0$, then $\frac{\partial f}{\partial y}=xy^{x-1}[(1+\frac{1}{y})^{x-1}-p^x]\leq0$. Thus, $f(x,y)$ is a decreasing function of $y$, i.e, $f(x,y)\geq f(x,\frac{1}{k})=\frac{(1+k)^x-p^x}{k^x}$. Set $y=\frac{1}{t}$, $t\geq k$, we get $f(x,\frac{1}{t})=\frac{(1+t)^x-p^x}{t^x}$. Thus, $(1+t)^x\geq p^x+\frac{(1+k)^x-p^x}{k^x}t^x$.
\end{proof}

\begin{thm}\label{thm:1}
For any $2\otimes 2\otimes2^{n-2}$ tripartite state $\rho_{A_1A_2A_3}\in H_{A_1}\otimes H_{A_2}\otimes H_{A_3}$ and $k\geq1$, $\frac{1}{2}\leq p\leq1$, $0\leq\alpha\leq\frac{r}{2}$, $r\geq2$,

(1) if $C^r(\rho_{A_1A_3})\geq kC^r(\rho_{A_1A_2})$, the concurrence satisfies
\begin{eqnarray}
C^\alpha(\rho_{A_1|A_2A_3})\geq p^{\frac{\alpha}{r}}C^\alpha(\rho_{A_1A_2})+\frac{(1+k)^{\frac{\alpha}{r}}-p^{\frac{\alpha}{r}}}{k^{\frac{\alpha}{r}}}C^\alpha(\rho_{A_1A_3}).
\end{eqnarray}

(2) if $C^r(\rho_{A_1A_2})\geq kC^r(\rho_{A_1A_3})$, the concurrence satisfies
\begin{eqnarray}
C^\alpha(\rho_{A_1|A_2A_3})\geq p^{\frac{\alpha}{r}}C^\alpha(\rho_{A_1A_3})+\frac{(1+k)^{\frac{\alpha}{r}}-p^{\frac{\alpha}{r}}}{k^{\frac{\alpha}{r}}}C^\alpha(\rho_{A_1A_2}).
\end{eqnarray}
\end{thm}
\begin{proof}
For any $2\otimes 2\otimes2^{n-2}$ tripartite state $\rho_{A_1A_2A_3}$ and $r\geq2$, one has that [6]:
\begin{eqnarray}
C^r(\rho_{A_1|A_2A_3})\geq C^r(\rho_{A_1A_2})+C^r(\rho_{A_1A_3}).
\end{eqnarray}
Obviously the inequality (6) holds if one of $C(\rho_{A_1A_3})$ or $C(\rho_{A_1A_2})$ vanishes.
Assuming $C^r(\rho_{A_1A_3})\geq kC^r(\rho_{A_1A_2})>0$, we have that
\begin{eqnarray}
C^{\alpha}(\rho_{A_1|A_2A_3})&\geq&(C^r(\rho_{A_1A_2})+C^r(\rho_{A_1A_3}))^x\nonumber\\
&=&C^{rx}(\rho_{A_1A_2})(1+\frac{C^r(\rho_{A_1A_3})}{C^r(\rho_{A_1A_2})})^x\nonumber\\
&\geq& C^{rx}(\rho_{A_1A_2})[p^x+\frac{(1+k)^x-p^x}{k^x}(\frac{C^r(\rho_{A_1A_3})}{C^r(\rho_{A_1A_2})})^x]\nonumber\\
&=&p^{\frac{\alpha}{r}}C^{\alpha}(\rho_{A_1A_2})+\frac{(1+k)^{\frac{\alpha}{r}}-p^{\frac{\alpha}{r}}}{k^{\frac{\alpha}{r}}}C^{\alpha}(\rho_{A_1A_3}),
\end{eqnarray}
where $\alpha=rx$, $0\leq\alpha\leq\frac{r}{2}$ as $0\leq x\leq\frac{1}{2}$, and the second inequality is due to (5). Using similar methods, we obtain the inequality (7).
\end{proof}

For simplicity, we denote $C(\rho_{A_1A_i})$ $(i=2,\cdots,n-1)$ by $C_{A_1A_i}$, $C(\rho_{A_1|A_{j+1}\cdots A_n})$ $(j=1,\cdots,n-1)$ by $C_{A_1|A_{j+1}\cdots A_n}$, and $l=\frac{(1+k)^{\frac{\alpha}{r}}-p^{\frac{\alpha}{r}}}{k^{\frac{\alpha}{r}}}$ ($k\geq1$, $\frac{1}{2}\leq p\leq1$, $0\leq\alpha\leq\frac{r}{2}$, $r\geq2$). The monogamy inequality of the $\alpha$th power of concurrence for $n$-qubit quantum states is given by the following theorem for $0\leq\alpha\leq\frac{r}{2}$ and $r\geq2$.
\begin{thm}\label{thm:2}
For any $n$-qubit quantum state $\rho_{A_1A_2A_3\cdots A_n}$ and $k\geq1$, $\frac{1}{2}\leq p\leq1$, $0\leq\alpha\leq\frac{r}{2}$, $r\geq2$, we have that

(1) if $kC_{A_1A_i}^r\leq C_{A_1|A_{i+1}\cdots A_n}^r$ for $i=2,\cdots,m$ and $C_{A_1A_j}^r\geq kC_{A_1|A_{j+1}\cdots A_n}^r$ for $j=m+1,\cdots,n-1$, $\forall$ $2\leq m\leq n-2$, $n\geq4$, we have
\begin{eqnarray}
C_{A_1|A_2A_3\cdots A_n}^\alpha&\geq&p^{\frac{\alpha}{r}}(C_{A_1A_2}^\alpha+lC_{A_1A_3}^\alpha+\cdots+l^{m-2}C_{A_1A_m}^\alpha)\nonumber\\
&+&l^m[C_{A_1A_{m+1}}^\alpha+p^{\frac{\alpha}{r}}C_{A_1A_{m+2}}^\alpha+\cdots+p^{\frac{(n-m-2)\alpha}{r}}C_{A_1A_{n-1}}^\alpha]\nonumber\\
&+&l^{m-1}p^{\frac{(n-m-1)\alpha}{r}}C_{A_1A_n}^\alpha.
\end{eqnarray}

(2) if $kC_{A_1A_i}^r\leq C_{A_1|A_{i+1}\cdots A_n}^r$ for $i=2,\cdots,n-1$ and $n\geq3$, we have
\begin{eqnarray}
C_{A_1|A_2A_3\cdots A_n}^\alpha\geq p^{\frac{\alpha}{r}}(C_{A_1A_2}^\alpha+lC_{A_1A_3}^\alpha+\cdots+l^{n-3}C_{A_1A_{n-1}}^\alpha)+l^{n-2}C_{A_1A_n}^\alpha.
\end{eqnarray}

(3) if $C_{A_1A_j}^r\geq kC_{A_1|A_{j+1}\cdots A_n}^r$ for $j=2,\cdots,n-1$ and $n\geq3$, we have
\begin{eqnarray}
C_{A_1|A_2A_3\cdots A_n}^\alpha\geq l(C_{A_1A_2}^\alpha+p^{\frac{\alpha}{r}}C_{A_1A_3}^\alpha+\cdots+p^{\frac{(n-3)\alpha}{r}}C_{A_1A_{n-1}}^\alpha)+p^{\frac{(n-2)\alpha}{r}}C_{A_1A_n}^\alpha.
\end{eqnarray}
\end{thm}

\begin{proof}
For the $n$-qubit quantum state $\rho_{A_1A_2A_3\cdots A_n}$, using inequality (6) repeatedly if $kC_{A_1A_i}^r\leq C_{A_1|A_{i+1}\cdots A_n}^r$ for $i=2,\cdots,m$, we have
\begin{eqnarray}
C_{A_1|A_2A_3\cdots A_n}^\alpha&\geq& p^{\frac{\alpha}{r}}C_{A_1A_2}^\alpha+lC_{A_1|A_3\cdots A_n}^\alpha \nonumber\\
&\geq& p^{\frac{\alpha}{r}}C_{A_1A_2}^\alpha+lp^{\frac{\alpha}{r}}C_{A_1A_3}^\alpha+l^2C_{A_1|A_4\cdots A_n}^\alpha \nonumber\\
&\geq&\cdots\nonumber\\
&\geq& p^{\frac{\alpha}{r}}(C_{A_1A_2}^\alpha+lC_{A_1A_3}^\alpha+\cdots+l^{m-2}C_{A_1A_m}^\alpha)+l^{m-1}C_{A_1|A_{m+1}\cdots A_n}^\alpha.
\end{eqnarray}
 From the inequality (7) and $C_{A_1A_j}^r\geq kC_{A_1|A_{j+1}\cdots A_n}^r$ for $j=m+1,\cdots,n-1$, we get
\begin{eqnarray}
 C_{A_1|A_{m+1}\cdots A_n}^\alpha&\geq& lC_{A_1A_{m+1}}^\alpha+p^{\frac{\alpha}{r}}C_{A_1|A_{m+2}\cdots A_n}^\alpha \nonumber\\
&\geq&\cdots \nonumber\\
&\geq& l[C_{A_1A_{m+1}}^\alpha+p^{\frac{\alpha}{r}}C_{A_1A_{m+2}}^\alpha+\cdots+p^{\frac{(n-m-2)\alpha}{r}}C_{A_1A_{n-1}}^\alpha]
+p^{\frac{(n-m-1)\alpha}{r}}C_{A_1A_n}^\alpha.
\end{eqnarray}
Combining (13) and (14), one gets (10). If all  $kC_{A_1A_i}^r\leq C_{A_1|A_{i+1}\cdots A_n}^r$ for $i=2,\cdots,n-1$ or $C_{A_1A_j}^r\geq kC_{A_1|A_{j+1}\cdots A_n}^r$ for $j=2,\cdots,n-1$, we have the inequality (11) and (12).

\end{proof}

{\bf Remark 1.} We provide a class of lower bounds of the $\alpha$th power of concurrence with different $p$. For $\frac{1}{2}\leq p\leq1$, the lower bound of the $\alpha$th power of concurrence is tighter when $p$ is smaller. When $p=1$ and $C^r(\rho_{A_1A_3})\geq kC^r(\rho_{A_1A_2})$, we have $C^\alpha(\rho_{A_1|A_2A_3})\geq C^\alpha(\rho_{A_1A_2})+\frac{(1+k)^{\frac{\alpha}{r}}-1}{k^{\frac{\alpha}{r}}}C^\alpha(\rho_{A_1A_3})$ for $0\leq\alpha\leq\frac{r}{2}$ and $r\geq2$. Using concurrence as the quantum correlation measure, our Theorem 1 is a generalization of Theorem 1 in [15].

{\bf Remark 2.} Taking the tripartite quantum states as an example, when $C^r(\rho_{A_1A_3})\geq kC^r(\rho_{A_1A_2})$, one has that $C^\alpha(\rho_{A_1|A_2A_3})\geq C^\alpha(\rho_{A_1A_2})+\frac{(1+k)^{\frac{\alpha}{r}}-1}{k^{\frac{\alpha}{r}}}C^\alpha(\rho_{A_1A_3})=z_1$ in [15]. In Theorem 1, the $\alpha$th power of concurrence satisfies $C^\alpha(\rho_{A_1|A_2A_3})\geq p^{\frac{\alpha}{r}}C^\alpha(\rho_{A_1A_2})+\frac{(1+k)^{\frac{\alpha}{r}}-p^{\frac{\alpha}{r}}}{k^{\frac{\alpha}{r}}}C^\alpha(\rho_{A_1A_3})=z_2$. Let $z=z_2-z_1$, we have $z\geq0$ for $\frac{1}{2}\leq p\leq1$, $0\leq\alpha\leq\frac{r}{2}$ and $r\geq2$, then our results are seen tighter than that in [15].

\textit{\textbf{Example 1.}} Consider the qubit 3-state $\rho=|\varphi\rangle\langle\varphi|$, where $|\varphi\rangle$ is in the generalized Schmidt decomposition form [22]:
\begin{eqnarray}
|\varphi\rangle=\lambda_0|000\rangle+\lambda_1e^{i\theta}|100\rangle+\lambda_2|101\rangle+\lambda_3|110\rangle+\lambda_4|111\rangle,
\end{eqnarray}
with $0\leq\theta\leq\pi$, $\lambda_i\geq0$, $i=0,\cdots,4$ and $\sum_{i=0}^4\lambda_i^2=1$. We have $C(\rho_{A_1|A_2A_3})=2\lambda_0\sqrt{\lambda_2^2+\lambda_3^2+\lambda_4^2}$, $C(\rho_{A_1A_2})=2\lambda_0\lambda_2$, and $C(\rho_{A_1A_3})=2\lambda_0\lambda_3$.
Setting $\lambda_0=\lambda_2=\frac{1}{2}$, $\lambda_3=\frac{\sqrt{2}}{2}$, and $\lambda_1=\lambda_4=0$, then $C(\rho_{A_1|A_2A_3})=\frac{\sqrt{3}}{2}$, $C(\rho_{A_1A_2})=\frac{1}{2}$, $C(\rho_{A_1A_3})=\frac{\sqrt{2}}{2}$. Thus, $C^\alpha(\rho_{A_1|A_2A_3})=(\frac{\sqrt{3}}{2})^\alpha$. Let $p=\frac{1}{2}$ and $k=\sqrt{2}$. By Theorem 1, we have $C^\alpha(\rho_{A_1|A_2A_3})\geq (\frac{1}{2})^{\frac{\alpha}{r}}(\frac{1}{2})^\alpha+\frac{(1+\sqrt{2})^{\frac{\alpha}{r}}-(\frac{1}{2})^{\frac{\alpha}{r}}}{(\sqrt{2})^{\frac{\alpha}{r}}}(\frac{\sqrt{2}}{2})^\alpha=z_1$. The result of [15] says that that $C^\alpha(\rho_{A_1|A_2A_3})\geq (\frac{1}{2})^\alpha+\frac{(1+\sqrt{2})^{\frac{\alpha}{r}}-1}{(\sqrt{2})^{\frac{\alpha}{r}}}(\frac{\sqrt{2}}{2})^\alpha=z_2$.
One then sees that Theorem 1 gives a better result than that of Ref. [15] from Fig. 1.
In fact, let $z'=z_1-z_2=((\frac{1}{2})^{\frac{\alpha}{r}}-1)(\frac{1}{2})^\alpha+\frac{1-(\frac{1}{2})^{\frac{\alpha}{r}}}{(\sqrt{2})^{\frac{\alpha}{r}}}(\frac{\sqrt{2}}{2})^\alpha$. For $0\leq\alpha\leq1$ and $r\geq2$, we have $z'\geq0$, which is shown in Fig. 2.
\begin{figure}[!htb]
\centerline{\includegraphics[width=0.6\textwidth]{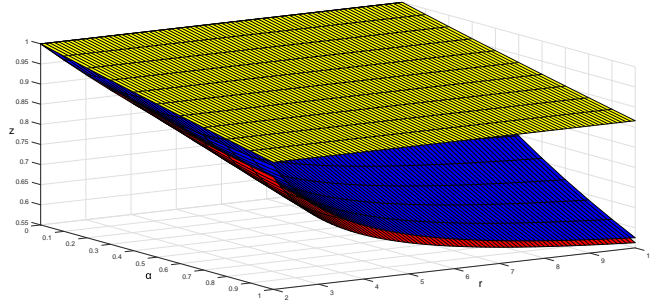}}
\renewcommand{\figurename}{Fig.}
\caption{The axis $z$ is the lower bounds of concurrence or concurrence itself of $|\varphi\rangle$. The yellow surface is the concurrence of the quantum state $|\varphi\rangle$. The lower bound of the concurrence in Ref. [15] is depicted in the red surface and the blue one is our result from Theorem 1.}
\end{figure}

\begin{figure}[!htb]
\centerline{\includegraphics[width=0.6\textwidth]{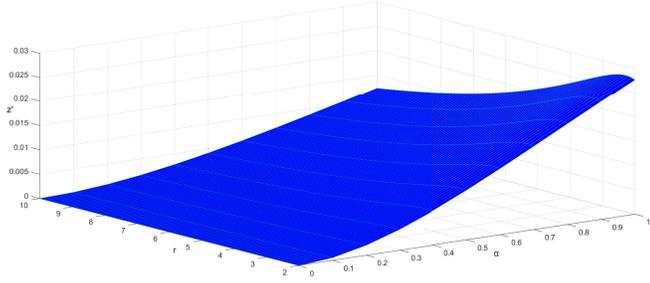}}
\renewcommand{\figurename}{Fig.}
\caption{The blue surface is the difference value $z'$ between the lower bound of concurrence $z_1$ from Theorem 1 and that of in Ref. [15].}
\end{figure}

Let $\tau$ be a bipartite quantum correlation measure. It is known that the measure satisfies [15, 23]:
\begin{eqnarray}
\tau(\rho_{A_1|A_2A_3\cdots A_n})\geq\tau(\rho_{A_1A_2})+\tau(\rho_{A_1A_3})+\cdots+\tau(\rho_{A_1A_n}),
\end{eqnarray}
where $\rho_{A_1|A_2A_3\cdots A_n}$ is the bipartite state under bipartite partition, and $\rho_{A_1A_i}$, $i=2,\cdots,n$, are the reduced density matrix of $\rho_{A_1A_2A_3\cdots A_n}$, i.e, $\rho_{A_1A_i}=tr_{A_2\cdots A_{i-1}A_{i+1}\cdots A_n}(\rho_{A_1A_2A_3\cdots A_n})$.

In addition to the squared concurrence, there are quantum correlation measures that also satisfy the monogamy relation, such as entanglement of formation and squared convex-roof extended negativity. Let $\tau$ be one of these quantum correlation measures, we have the following results:

\begin{corollary}
For any $n$-partite quantum state $\rho_{A_1A_2A_3\cdots A_n}$ and $k\geq1$, $\frac{1}{2}\leq p\leq1$, $0\leq\alpha\leq\frac{r}{2}$, $r\geq2$, we have that

(1) if $k\tau_{A_1A_i}^r\leq \tau_{A_1|A_{i+1}\cdots A_n}^r$ for $i=2,\cdots,m$ and $\tau_{A_1A_j}^r\geq k\tau_{A_1|A_{j+1}\cdots A_n}^r$ for $j=m+1,\cdots,n-1$, $\forall$ $2\leq m\leq n-2$, $n\geq4$, we have
\begin{eqnarray}
\tau_{A_1|A_2A_3\cdots A_n}^\alpha&\geq& p^{\frac{\alpha}{r}}(\tau_{A_1A_2}^\alpha+l\tau_{A_1A_3}^\alpha+\cdots+l^{m-2}\tau_{A_1A_m}^\alpha)\nonumber\\
&+&l^m[\tau_{A_1A_{m+1}}^\alpha+p^{\frac{\alpha}{r}}\tau_{A_1A_{m+2}}^\alpha+\cdots+p^{\frac{(n-m-2)\alpha}{r}}\tau_{A_1A_{n-1}}^\alpha]\nonumber\\
&+&l^{m-1}p^{\frac{(n-m-1)\alpha}{r}}\tau_{A_1A_n}^\alpha.
\end{eqnarray}

(2) if $k\tau_{A_1A_i}^r\leq \tau_{A_1|A_{i+1}\cdots A_n}^r$ for $i=2,\cdots,n-1$ and $n\geq3$, we have
\begin{eqnarray}
\tau_{A_1|A_2A_3\cdots A_n}^\alpha\geq p^{\frac{\alpha}{r}}(\tau_{A_1A_2}^\alpha+l\tau_{A_1A_3}^\alpha+\cdots+l^{n-3}\tau_{A_1A_{n-1}}^\alpha)+l^{n-2}\tau_{A_1A_n}^\alpha.
\end{eqnarray}

(3) if $\tau_{A_1A_j}^r\geq k\tau_{A_1|A_{j+1}\cdots A_n}^r$ for $j=2,\cdots,n-1$ and $n\geq3$, we have
\begin{eqnarray}
\tau_{A_1|A_2A_3\cdots A_n}^\alpha\geq l(\tau_{A_1A_2}^\alpha+p^{\frac{\alpha}{r}}\tau_{A_1A_3}^\alpha+\cdots+p^{\frac{(n-3)\alpha}{r}}\tau_{A_1A_{n-1}}^\alpha)+p^{\frac{(n-2)\alpha}{r}}\tau_{A_1A_n}^\alpha.
\end{eqnarray}

\end{corollary}

\section{Tighter polygamy relations of quantum correlations}
In this section, we derive the polygamy inequalities for multipartite quantum systems. We first give the polygamy inequalities of the $\beta$th power of concurrence for $\beta\leq0$. Next the polygamy inequalities of $\beta$th power of negativity for $\beta\geq s$ and $0<s\leq1$ are obtained. We also generalize the results to general quantum correlation measures.

\begin{lemma}\label{lemma:2}
For $0<q\leq1$, $k\geq1$ and $t\geq k$,
if $x\geq1$, or $x\leq0$, we have
\begin{eqnarray}
(1+t)^x\leq q^x+\frac{(1+k)^x-q^x}{k^x}t^x.
\end{eqnarray}
\end{lemma}
\begin{proof}
Let $f(x,y)=(1+y)^x-(qy)^x$ with $0<q\leq1$, $x\geq1$ and $0<y\leq\frac{1}{k}$. Then $\frac{\partial f}{\partial y}=xy^{x-1}[(1+\frac{1}{y})^{x-1}-q^x]>0$. Thus, $f(x,y)$ is an increasing function of $y$, i.e, $f(x,y)\leq f(x,\frac{1}{k})=\frac{(1+k)^x-q^x}{k^x}$. Set $y=\frac{1}{t}$, $t\geq k$, we get $f(x,\frac{1}{t})=\frac{(1+t)^x-q^x}{t^x}$. Thus, $(1+t)^x\leq q^x+\frac{(1+k)^x-q^x}{k^x}t^x$. The case of
$x\leq0$ can be verified similarly as above. 
\end{proof}

Using Lemma 2, we obtain the following results:
\begin{thm}\label{thm:3}
For any $2\otimes 2\otimes2^{n-2}$ tripartite state $\rho_{ABC}\in H_A\otimes H_B\otimes H_C$ and $k\geq1$, $0<q\leq1$, $\beta\leq0$, $s\geq2$,

(1) if $C^s(\rho_{A_1A_3})\geq kC^s(\rho_{A_1A_2})$, the concurrence satisfies
\begin{eqnarray}
C^\beta(\rho_{A_1|A_2A_3})\leq q^{\frac{\beta}{s}}C^\beta(\rho_{A_1A_2})+\frac{(1+k)^{\frac{\beta}{s}}-q^{\frac{\beta}{s}}}{k^{\frac{\beta}{s}}}C^\beta(\rho_{A_1A_3}).
\end{eqnarray}

(2) if $C^s(\rho_{A_1A_2})\geq kC^s(\rho_{A_1A_3})$, the concurrence satisfies
\begin{eqnarray}
C^\beta(\rho_{A_1|A_2A_3})\leq q^{\frac{\beta}{s}}C^\beta(\rho_{A_1A_3})+\frac{(1+k)^{\frac{\beta}{s}}-q^{\frac{\beta}{s}}}{k^{\frac{\beta}{s}}}C^\beta(\rho_{A_1A_2}).
\end{eqnarray}
\end{thm}
\begin{proof}
 Recalling [6], we have that $C^s(\rho_{A_1|A_2A_3})\geq C^s(\rho_{A_1A_2})+C^s(\rho_{A_1A_3})$ for arbitrary $2\otimes 2\otimes2^{n-2}$ tripartite state $\rho_{A_1A_2A_3}$ and $s\geq2$. The inequalities (21) are trivial when either  $C(\rho_{A_1A_3})=0$ or $C(\rho_{A_1A_2})=0$.
 Now assuming $C^s(\rho_{A_1A_3})\geq kC^s(\rho_{A_1A_2})>0$ and noting that $\beta=sx\leq0$ for some $x\leq0$, we have that
\begin{eqnarray}
C^{\beta}(\rho_{A_1|A_2A_3})&\leq&(C^s(\rho_{A_1A_2})+C^s(\rho_{A_1A_3}))^x\nonumber\\
&=&C^{sx}(\rho_{A_1A_2})(1+\frac{C^s(\rho_{A_1A_3})}{C^s(\rho_{A_1A_2})})^x\nonumber\\
&\leq& C^{sx}(\rho_{A_1A_2})[q^x+\frac{(1+k)^x-q^x}{k^x}(\frac{C^s(\rho_{A_1A_3})}{C^s(\rho_{A_1A_2})})^x]\nonumber\\
&=&q^{\frac{\beta}{s}}C^{\beta}(\rho_{A_1A_2})+\frac{(1+k)^{\frac{\beta}{s}}-q^{\frac{\beta}{s}}}{k^{\frac{\beta}{s}}}C^{\beta}(\rho_{A_1A_3}),
\end{eqnarray}
where the first inequality is due to $C^s(\rho_{A_1|A_2A_3})\geq C^s(\rho_{A_1A_2})+C^s(\rho_{A_1A_3})$ for $s\geq2$ and $x\leq0$. Using Lemma 2, we get the second inequality. Similar proof gives the inequality (22) by using Lemma 2.

\end{proof}

{\bf Remark 3.} For $0<q\leq1$, the lower bound of the $\beta$th power of concurrence with $\beta\leq0$ is tighter when $p$ is smaller.

\textit{\textbf{Example 2.}}  Let us consider again the quantum state (15). Let $\lambda_0=\lambda_3=\frac{1}{2}$, $\lambda_1=\lambda_2=\lambda_4=\frac{\sqrt{6}}{6}$, $q=\frac{1}{2}$ and $k=\frac{\sqrt{6}}{2}$, we get $C(\rho_{A_1|A_2A_3})=\frac{\sqrt{21}}{6}$, $C(\rho_{A_1A_2})=\frac{\sqrt{6}}{6}$, $C(\rho_{A_1A_3})=\frac{1}{2}$. Thus, $C^\beta(\rho_{A_1|A_2A_3})=(\frac{\sqrt{21}}{6})^\beta$. By Theorem 3, the lower bounds of $C^\beta(\rho_{A_1|A_2A_3})$ is $z_1=(\frac{1}{2})^{\frac{\alpha}{r}}(\frac{\sqrt{6}}{6})^\alpha+\frac{(1+\frac{\sqrt{6}}{2})^{\frac{\alpha}{r}}-(\frac{1}{2})^{\frac{\alpha}{r}}}{(\frac{\sqrt{6}}{2})^{\frac{\alpha}{r}}}(\frac{1}{2})^\alpha$. From Ref. [7], we have $z_2=\frac{1}{2}((\frac{\sqrt{6}}{6})^\alpha+(\frac{1}{2})^\alpha)$.
For $\alpha\leq0$ and $2\leq r\leq5$, one can see from Fig. 3 that Theorem 3 give stronger result than that of Ref. [7]. 
\begin{figure}[!htb]
\centerline{\includegraphics[width=0.6\textwidth]{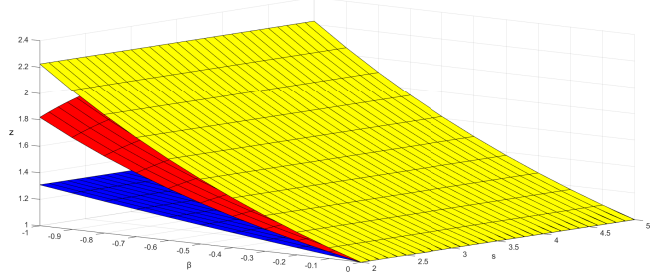}}
\renewcommand{\figurename}{Fig.}
\caption{The axis $z$ is the upper bounds of concurrence or concurrence itself of $|\varphi\rangle$. The blue surface is the concurrence of quantum state $|\varphi\rangle$. The lower bound of concurrence in Ref. [7] is the yellow surface and the red surface is our result from Theorem 3.}
\end{figure}

Let $H_{A_1}$ and $H_{A_2}$ be Hilbert spaces. The negativity of a quantum state $\rho_{A_1A_2}\in H_{A_1}\otimes H_{A_2}$ is defined by [24, 25],
\begin{eqnarray}
\mathcal{N}(\rho_{A_1A_2})=\|\rho_{A_1A_2}^{T_{A_1}}\|_{tr}-1,
\end{eqnarray}
where $\rho_{A_1A_2}^{T_{A_1}}$ denotes the partial transpose of $\rho_{A_1A_2}$ with respect to the subsystem $A_1$, and $\|\cdot\|_{tr}$ denotes the trace norm. The trace norm is defined as the sum of the singular values of the matrix $M\in\mathbb{R}^{m\times n}$, i.e., $\|M\|_{tr}=\sum_i\sigma_i=tr\sqrt{A^\dag A}$, where $\sigma_i$, $i=1,\cdots,min(m,n)$, are the singular values of the matrix $A$ arranged in descending order.

For any bipartite quantum pure state $|\varphi\rangle_{A_1A_2}$ with Schmidt rank $d$, $|\varphi\rangle_{A_1A_2}=\sum_{i=1}^d\sqrt{\lambda_i}|ii\rangle$, its negativity is given by [26],
\begin{eqnarray}
\mathcal{N}(|\varphi\rangle_{A_1A_2})=2\sum_{i<j}\lambda_i\lambda_j=(tr\sqrt{\rho_{A_1}})^2-1,
\end{eqnarray}
where $\lambda_i, i=1,\cdots,d$ are the eigenvalues for the reduced density matrix $\rho_{A_1}$ and $\rho_{A_1}=tr_{A_2}(|\varphi\rangle_{A_1A_2})$.

For a mixed bipartite quantum state $\rho_{A_1A_2}=\sum_ip_i|\varphi_i\rangle\langle\varphi_i|\in H_{A_1}\otimes H_{A_2}$, the negativity is given by the convex roof [27],
\begin{eqnarray}
\mathcal{N}(\rho_{A_1A_2})=min_{\{p_i,|\varphi_i\rangle\}}\sum_ip_i\mathcal{N}(|\varphi_i\rangle),
\end{eqnarray}
where the minimum is taken over all possible convex partitions of $\rho_{A_1A_2}$ into pure state ensembles $\{p_i,|\varphi_i\rangle\}$, $0\leq p_i\leq1$ and $\sum_ip_i=1$.

The convex-roof extended negativity of assistance for mixed states is defined by [28],
\begin{eqnarray}
\mathcal{N}_a(\rho_{A_1A_2})=max_{\{p_i,|\varphi_i\rangle\}}\sum_ip_i\mathcal{N}(|\varphi_i\rangle),
\end{eqnarray}
where the maximum is taken over all possible convex partitions of $\rho_{A_1A_2}$ into pure state ensembles $\{p_i,|\varphi_i\rangle\}$, $0\leq p_i\leq1$ and $\sum_ip_i=1$.

For the $n$-qubit quantum state $\rho_{A_1|A_2A_3\cdots A_n}$, regarded as a bipartite state under bipartite partition $A_1|A_2A_3\cdots A_n$, the negativity satisfies for $0\leq\beta\leq1$ [28]:
\begin{eqnarray}
\mathcal{N}_a^\beta(\rho_{A_1|A_2A_3\cdots A_n})\leq \sum_{j=2}^{n}\beta^j\mathcal{N}_a^\beta(\rho_{A_1A_j})\leq\sum_{j=2}^{n}\mathcal{N}_a^\beta(\rho_{A_1A_j}),
\end{eqnarray}
where $\rho_{A_1A_i}$, $i=2,\cdots,n$, is the reduced density matrix of $\rho$, i.e, $\rho_{A_1A_i}=tr_{A_2\cdots A_{i-1}A_{i+1}\cdots A_n}(\rho)$.

Using Lemma 2 and the same idea of Theorem 3, we get the following polygamy inequalities of the $\beta$th power of negativity for $\beta\geq s$, $0<s\leq1$.
\begin{thm}\label{thm:4}
For any $2\otimes 2\otimes2^{n-2}$ tripartite state $\rho_{A_1A_2A_3}\in H_{A_1}\otimes H_{A_2}\otimes H_{A_3}$ and $k\geq1$, $0<q\leq1$, $\beta\geq s$, $0<s\leq1$.

(1) if ${\mathcal{N}_a}^s(\rho_{A_1A_3})\geq k{\mathcal{N}_a}^s(\rho_{A_1A_2})$, the concurrence satisfies
\begin{eqnarray}
{\mathcal{N}_a}^\beta(\rho_{A_1|A_2A_3})\leq q^{\frac{\beta}{s}}{\mathcal{N}_a}^\beta(\rho_{A_1A_2})+\frac{(1+k)^{\frac{\beta}{s}}-q^{\frac{\beta}{s}}}{k^{\frac{\beta}{s}}}{\mathcal{N}_a}^\beta(\rho_{A_1A_3}).
\end{eqnarray}

(2) if ${\mathcal{N}_a}^s(\rho_{A_1A_2})\geq k{\mathcal{N}_a}^s(\rho_{A_1A_3})$, the concurrence satisfies
\begin{eqnarray}
{\mathcal{N}_a}^\beta(\rho_{A_1|A_2A_3})\leq q^{\frac{\beta}{r}}{\mathcal{N}_a}^\beta(\rho_{A_1A_3})+\frac{(1+k)^{\frac{\beta}{s}}-q^{\frac{\beta}{s}}}{k^{\frac{\beta}{s}}}{\mathcal{N}_a}^\beta(\rho_{A_1A_2}).
\end{eqnarray}
\end{thm}

Let ${\mathcal{N}_a}(\rho_{A_1A_i})$ $(i=2,\cdots,n-1)$ by ${\mathcal{N}_a}_{A_1A_i}$, ${\mathcal{N}_a}(\rho_{A_1|A_{j+1}\cdots A_n})$ $(j=1,\cdots,n-1)$ by ${\mathcal{N}_a}_{A_1|A_{j+1}\cdots A_n}$, and $l=\frac{(1+k)^{\frac{\alpha}{r}}-q^{\frac{\alpha}{r}}}{k^{\frac{\alpha}{r}}}$ ($k\geq1$, $0<q\leq1$, $\beta\geq s$, $0<s\leq1$). The polygamy inequality of negativity for $n$-qubit quantum states is given by the following result:

\begin{thm} \label{thm:5}
For any $n$-qubit quantum state $\rho_{A_1A_2A_3\cdots A_n}$ and $k\geq1$, $0<q\leq1$, $\beta\geq s$, $0<s\leq1$, we obtain that

(1) if $k{\mathcal{N}_a}_{A_1A_i}^s\leq {\mathcal{N}_a}_{A_1|A_{i+1}\cdots A_n}^s$ for $i=2,\cdots,m$ and ${\mathcal{N}_a}_{A_1A_j}^s\geq k\mathcal{N}_{A_1|A_{j+1}\cdots A_n}^s$ for $j=m+1,\cdots,n-1$, $\forall$ $2\leq m\leq n-2$, $n\geq4$, we have
\begin{eqnarray}
{\mathcal{N}_a}_{A_1|A_2A_3\cdots A_n}^\beta&\leq& q^{\frac{\beta}{s}}({\mathcal{N}_a}_{A_1A_2}^\beta+l{\mathcal{N}_a}_{A_1A_3}^\beta+\cdots+l^{m-2}{\mathcal{N}_a}_{A_1A_m}^\beta)\nonumber\\
&+&l^m[{\mathcal{N}_a}_{A_1A_{m+1}}^\beta+q^{\frac{\beta}{s}}{\mathcal{N}_a}_{A_1A_{m+2}}^\beta+\cdots+q^{\frac{(n-m-2)\beta}{s}}{\mathcal{N}_a}_{A_1A_{n-1}}^\beta]\nonumber\\
&+&l^{m-1}q^{\frac{(n-m-1)\beta}{s}}{\mathcal{N}_a}_{A_1A_n}^\beta.
\end{eqnarray}

(2) if $k{\mathcal{N}_a}_{A_1A_i}^s\leq {\mathcal{N}_a}_{A_1|A_{i+1}\cdots A_n}^s$ for $i=2,\cdots,n-1$ and $n\geq3$, we have
\begin{eqnarray}
{\mathcal{N}_a}_{A_1|A_2A_3\cdots A_n}^\beta\leq q^{\frac{\beta}{s}}({\mathcal{N}_a}_{A_1A_2}^\beta+l{\mathcal{N}_a}_{A_1A_3}^\beta+\cdots+l^{n-3}{\mathcal{N}_a}_{A_1A_{n-1}}^\beta)+l^{n-2}{\mathcal{N}_a}_{A_1A_n}^\beta.
\end{eqnarray}

(3) if ${\mathcal{N}_a}_{A_1A_j}^s\geq k{\mathcal{N}_a}_{A_1|A_{j+1}\cdots A_n}^s$ for $j=2,\cdots,n-1$ and $n\geq3$, we have
\begin{eqnarray}
{\mathcal{N}_a}_{A_1|A_2A_3\cdots A_n}^\beta&\leq& l[{\mathcal{N}_a}_{A_1A_2}^\beta+q^{\frac{\beta}{s}}{\mathcal{N}_a}_{A_1A_3}^\beta+\cdots+q^{\frac{(n-3)\beta}{s}}{\mathcal{N}_a}_{A_1A_{n-1}}^\beta]\nonumber\\
&+&q^{\frac{(n-2)\beta}{s}}{\mathcal{N}_a}_{A_1A_n}^\beta.
\end{eqnarray}

\end{thm}

We can generalize our results to polygamy inequality for quantum correlation measures. Let quantum correlation measure $\tau$ be
the negativity of assistance, concurrence of assistance or entanglement of assistance, then we have the following polygamy inequalities:

\begin{corollary}
For any $n$-partite quantum state $\rho_{A_1A_2A_3\cdots A_n}$, $\beta\geq s$, $0<s\leq1$ or $\beta\leq0$, we have that

(1) if $k\tau_{A_1A_i}^s\leq \tau_{A_1|A_{i+1}\cdots A_n}^s$ for $i=2,\cdots,m$ and $\tau_{A_1A_j}^s\geq k\tau_{A_1|A_{j+1}\cdots A_n}^s$ for $j=m+1,\cdots,n-1$, $\forall$ $2\leq m\leq n-2$, $n\geq4$, we have
\begin{eqnarray}
\tau_{A_1|A_2A_3\cdots A_n}^\beta&\leq& q^{\frac{\beta}{s}}(\tau_{A_1A_2}^\beta+l\tau_{A_1A_3}^\beta+\cdots+l^{m-2}\tau_{A_1A_m}^\beta)\nonumber\\
&+&l^m[\tau_{A_1A_{m+1}}^\beta+q^{\frac{\beta}{s}}\tau_{A_1A_{m+2}}^\beta+\cdots+q^{\frac{(n-m-2)\beta}{s}}\tau_{A_1A_{n-1}}^\beta]\nonumber\\
&+&l^{m-1}q^{\frac{(n-m-1)\beta}{s}}\tau_{A_1A_n}^\beta,
\end{eqnarray}
where $k\geq1$, and $0\leq q\leq1$.

(2) if $k\tau_{A_1A_i}^s\leq \tau_{A_1|A_{i+1}\cdots A_n}^s$ for $i=2,\cdots,n-1$ and $n\geq3$, we have that
\begin{eqnarray}
\tau_{A_1|A_2A_3\cdots A_n}^\beta\leq q^{\frac{\beta}{s}}(\tau_{A_1A_2}^\beta+l\tau_{A_1A_3}^\beta+\cdots+l^{n-3}\tau_{A_1A_{n-1}}^\beta)+l^{n-2}\tau_{A_1A_n}^\beta,
\end{eqnarray}
where $k\geq1$, and $0\leq q\leq1$.

(3) if $\tau_{A_1A_j}^s\geq k\tau_{A_1|A_{j+1}\cdots A_n}^s$ for $j=2,\cdots,n-1$ and $n\geq3$, we have
\begin{eqnarray}
\tau_{A_1|A_2A_3\cdots A_n}^\beta\leq l[\tau_{A_1A_2}^\beta+q^{\frac{\beta}{s}}\tau_{A_1A_3}^\beta+\cdots+q^{\frac{(n-3)\beta}{s}}\tau_{A_1A_{n-1}}^\beta]+q^{\frac{(n-2)\beta}{s}}\tau_{A_1A_n}^\beta,
\end{eqnarray}
where $k\geq1$, and $0\leq q\leq1$.
\end{corollary}

\section{Conclusion}
Mono- and polygamous inequalities 
of quantum correlations are fundamental properties in quantum systems. In this paper we have
studied several families of monogamy and polygamy inequalities of measures for quantum correlation in arbitrary dimensional multipartite quantum systems.
We have derived the monogamy inequality of the $\alpha$th ($0\leq\alpha\leq\frac{r}{2}, r\geq2$) power of the concurrence for $2\otimes2\otimes2^{n-2}$ quantum states. Similarly, analytical monogamy inequalities for the $n$-qubit states have also been presented. The same method is generalized to the monogamy relations of quantum correlation for multipartite quantum systems. These new inequalities provide new and better angles to study the
quantum systems. We have used examples to show that our results are tighter than the existing ones. Moreover, we have presented for the polygamy inequality of the $\beta$th ($\beta\leq0$) power of concurrence and the $\beta$th ($\beta\geq s, 0\leq s\leq1$) power of negativity for $2\otimes2\otimes2^{n-2}$ quantum states and generalized to the $n$-qubit quantum systems. Finally, we have derived the polygamy inequalities of other quantum correlation measures such as concurrence of assistance and entanglement of assistance.

\textbf {Acknowledgments}
This work is supported in part by the Simons Foundation under grant no. 523868 and the National Natural Science Foundation of China
under grant nos. 12126351 and 12126314.

\vskip 1in
\textbf {Statements and Declarations}

This work is supported in part by the Simons Foundation under grant no. 523868 and the National Natural Science Foundation of China
under grant nos. 12126351 and 12126314.

\end{CJK*}
\end{document}